\documentclass[sn-mathphys,Numbered]{sn-jnl}

\usepackage{graphicx}%
\usepackage{multirow}%
\usepackage{amsmath,amssymb,amsfonts}%
\usepackage{amsthm}%
\usepackage{mathrsfs}%
\usepackage[title]{appendix}%
\usepackage{xcolor}%
\usepackage{textcomp}%
\usepackage{manyfoot}%
\usepackage{booktabs}%
\usepackage{algorithm}%
\usepackage{algorithmicx}%
\usepackage{algpseudocode}%
\usepackage{listings}%
\usepackage{tikz}%
\newcounter{cases}
\newcounter{subcases}[cases]
\newenvironment{mycase}
{
	\setcounter{cases}{0}
	\setcounter{subcases}{0}
	\newcommand{\case}
	{
		\par\indent\stepcounter{cases}\textit{Case \thecases.}
	}
	
}
{
	\par
}
\renewcommand*\thecases{\arabic{cases}}

\theoremstyle{thmstyleone}%
\newtheorem{theorem}{Theorem}
\theoremstyle{thmstyletwo}%
\newtheorem{example}{Example}%
\newtheorem{remark}{Remark}%
\newtheorem{lemma}{Lemma}
\theoremstyle{thmstylethree}%
\newtheorem{definition}{Definition}%
\newcommand{\Mod}[1]{\ (\mathrm{mod}\ #1)}
\begin{document}
\title[Article Title]{A Construction of Arbitrarily Large Type-II $Z$ Complementary Code Set}


\author[1]{\fnm{Rajen} \sur{Kumar} \href{https://orcid.org/0000-0003-2100-6447}{\includegraphics[scale=0.01]{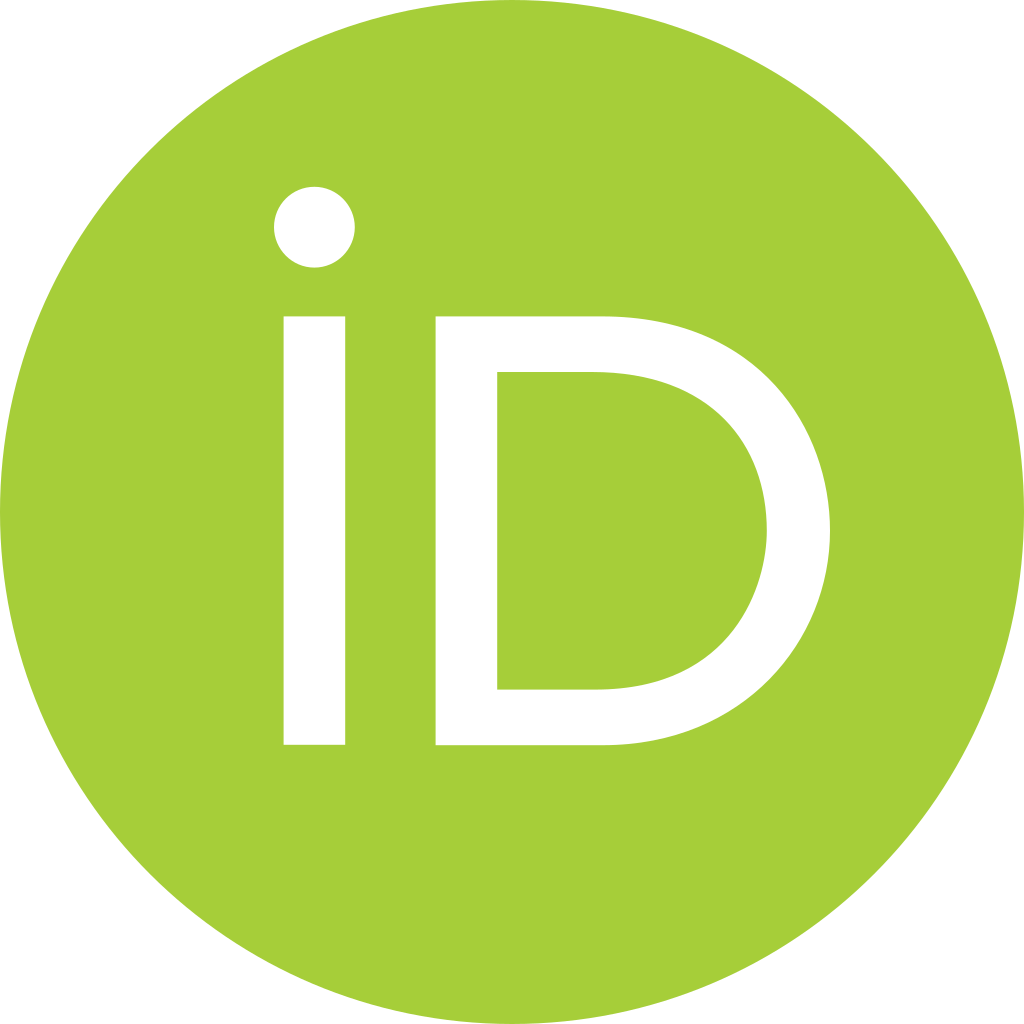}}}\email{rajen\_2021ma4@iitp.ac.in}

\author[1]{\fnm{Prashant Kumar} \sur{Srivastava} \href{https://orcid.org/0000-0002-7651-5639}{\includegraphics[scale=0.01]{orcid_logo.png}}}\email{pksri@iitp.ac.in}

\author*[2]{\fnm{Sudhan} \sur{Majhi} \href{https://orcid.org/0000-0002-2142-1862}{\includegraphics[scale=0.01]{orcid_logo.png}}}\email{smajhi@iisc.ac.in}

\affil[1]{\orgdiv{Department of Mathematics}, \orgname{Indian Institute of Technology}, \orgaddress{\city{Patna}, \postcode{801106}, \state{Bihar}, \country{India}}}

\affil[2]{\orgdiv{Department of Electrical Communication Engineering}, \orgname{Indian Institute of Sciene}, \orgaddress{ \city{Bangalore}, \postcode{560012}, \state{Karnataka}, \country{India}}}



\abstract{For a type-I $(K,M,Z,N)$-ZCCS, it follows $K \leq M \left\lfloor \frac{N}{Z}\right\rfloor$. In this paper, we propose a construction of type-II  $(p^{k+n},p^k,p^{n+r}-p^r+1,p^{n+r})$-$Z$ complementary code set (ZCCS) using an extended Boolean function, its properties of Hamiltonian paths and the concept of isolated vertices, where $p\ge 2$.  However, the proposed type-II ZCCS provides $K = M(N-Z+1)$ codes, where as for type-I $(K,M,N,Z)$-ZCCS, it is $K \leq M \left\lfloor \frac{N}{Z}\right\rfloor$. Therefore, the proposed type-II ZCCS provides a larger number of codes compared to type-I ZCCS. Further, as a special case of the proposed construction, $(p^k,p^k,p^n)$-CCC can be generated, for any integral value of $p\ge2$ and $k\le n$.}




\maketitle
\section{Introduction}
A collection of codes is described as a complete complementary code (CCC) if the aperiodic auto-correlation function (AACF) value of each code is zero for all non-zero time shifts and the aperiodic cross-correlation function (ACCF) value between any two codes is zero for all time shifts \cite{sue}. CCC has shown a significant role in asynchronous multi-carrier code division multiple access (MC-CDMA) for interference-free communication  \cite{nextGenCDMA}. Rathinakumar and Chaturvedi proposed $(2^{k+1},2^{k+1},2^{k+m})$-CCC using a generalized Boolean function (GBF) and by associating it with a graph, where a Hamiltonian path of $m$ vertices is obtained after deleting the $k$ vertices from the graph \cite{Rat}. Recently, construction of CCCs with new and flexible parameters has been proposed in \cite{Palash_ar} based on a multivariable function. To utilise more users in the MC-CDMA system, the zero correlation zone (ZCZ) has been introduced and formed $Z$-complementary code set (ZCCS) in \cite{Fan, Feng_ZCCS}. ZCCSs are capable of supporting interference-free MC-CDMA in a quasi-synchronous environment without the need for power control \cite{nextGenCDMA}. Depending on whether the ZCZ width exists close to zero time shift or end time shift, the code set is referred to as a type-I ZCCS or type-II ZCCS, respectively. A $(K,M,Z,N)$-ZCCS is a collection of $K$ matrices, where each matrix has $M$ constituent sequences of length $N$ with ZCZ width $Z$. The number of constituent sequences is referred to as flock size. For a type-I $(K,M,Z,N)$-ZCCS, $K \leq M \left\lfloor \frac{N}{Z}\right\rfloor$ \cite{Fan} and there are several such type-I ZCCS constructions provided in \cite{Li, Das, SarkarLCOM, Sarkarpseudo, Ghosh, Wu}. Recently, type-II ZCCS has been proposed to increase the number of codes and ZCZ width \cite{Rajen_ar}. As type-II ZCCS has a wider ZCZ width and more codes compared to type-I ZCCS, it is suitable for the Quasi-synchronous MC-CDMA system with delays \cite{Rajen_ar}.

In this paper, we propose a direct construction of type-II $(p^{r+k},p^k,p^{n+r}-p^r+1,p^{n+r})$-ZCCS, where $p \ge 2$, $0<k<n$ and $r\ge 0$, i.e., we have $p^r$ times more codes compared to the flock size by compromising the $p^r-1$ time shift of ZCZ width. Therefore, in the proposed construction, a large number of codes is generated compared to type-I ZCCS with the same flock size. An extended Boolean function (EBF) is used to construct a type-II ZCCS. The construction is associated with a graph of a collection of $k$ Hamiltonian paths and $r$ isolated vertices. The proposed method also generates $(p^k,p^k,p^n)$-CCC for $k\le n$, $p\ge 2$. The works present in \cite{Davis, Rat, Shen_SNC} are special \textcolor{black}{cases} of the proposed construction.

The remainder of this manuscript has been structured as follows. Section \ref{Sec:Pre} comprehensively describes the notations, important definitions, and relations between a graph and an EBF. In Section \ref{Sec:PC}, we propose a construction of type-II ZCCS and provide an example. In Section \ref{sec:compar}, a comparison is made between the proposed construction and existing work. Section \ref{Sec:Con} comprises the concluding remarks of the paper.
\section{Preliminaries}\label{Sec:Pre}
In this section, we mention a few important definitions and symbols that will be utilised throughout the paper.
\begin{definition}
	Let $\mathbf{a}=(a_0,a_1,\hdots,a_{N-1})$ and $\mathbf{b}=(b_{0},b_{1},\hdots,b_{N-1})$ be two complex-valued sequences of length $N$. The \textcolor{black}{aperiodic cross correlation function} (ACCF) between $\mathbf{a}$ and $\mathbf{b}$ for time shift $\tau$ is defined as
 \begin{equation}
     \mathcal{C}_{\mathbf{a}, \mathbf{b}}(\tau)=\left\{\begin{array}{ll}
		\sum_{i=0}^{N-1-\tau} a_{i+\tau}^1 b_{i}^{*}, & 0 \leq \tau<N, \\
		\sum_{i=0}^{N+\tau-1} a_{i}^1 b_{i-\tau}^{*}, & -N<\tau<0, \\
		0, & \text { otherwise, }
	\end{array}\right.
 \end{equation}
	where $b_{i}^{*}$ is complex conjugate of $b_{i}$. When $\mathbf{a}=\mathbf{b}$, $\mathcal{C}_{\mathbf{a},\mathbf{b}}(\tau)$ is called AACF of $\mathbf{a}$ and is denoted as $\mathcal{C}_{\mathbf{a}}(\tau)$.
\end{definition}
Let $C_{i}$ be a code as
	\begin{equation}\label{Eq:code}
C_{i}=\left[\begin{array}{c} \mathbf{a}_{0}^{i} \\ \mathbf{a}_{1}^{i} \\ \vdots \\ \mathbf{a}_{M-1}^{i} \end{array}\right]_{M \times N},
	\end{equation}
where, $\mathbf{a}^i_\nu$ is $(\nu+1)$th row of the matrix $C_i$ with $N$ elements.
 \begin{definition}
     Let $C_{s}$ and $C_{t} $ be any two codes defined in \eqref{Eq:code}, then ACCF between $C_{s}$ and $C_{t}$ is defined by
	\begin{equation}
\mathcal{C}\left({C}_{s},{C}_{t}\right)(\tau)=\sum_{\nu=0}^{M-1} \mathcal{C}\left(\mathbf{a}_{\nu}^{s}, \mathbf{a}_{\nu}^{t}\right)(\tau).
	\end{equation}
When ${C}_{s}={C}_{t}$, $\mathcal{C}(C_{s},C_{t})(\tau)$ is called AACF of ${C}_{s}$ and is denoted as $\mathcal{C}(C_{s})(\tau)$.
 \end{definition}
  \begin{definition}
Let $\mathbf{C}=\left\{{C}_{0},{C}_{1}, \ldots, {C}_{K-1}\right\}$ be a set of $K$ codes defined in \eqref{Eq:code}, then $\mathbf{C}$ is called a type-I $(K,M,Z,N)$-ZCCS if it satisfies the following properties
	\begin{equation}
	\mathcal{C}\left(C_{k_{1}}, C_{k_{2}}\right)(\tau)=\left\{\begin{array}{ll}
			NM, & \tau=0, k_{1}=k_{2}, \\
			0, &  \tau=0, k_{1} \neq k_{2},\\
			0, & 1\le |\tau|<Z.
		\end{array}\right.
	\end{equation}
\end{definition}
 \begin{definition}
Let $\mathbf{C}=\left\{{C}_{0},{C}_{1}, \ldots, {C}_{K-1}\right\}$ be a set of $K$codes defined in \eqref{Eq:code}, then $\mathbf{C}$ is called a type-II $(K,M,Z,N)$-ZCCS if it satisfies the following properties
	\begin{equation}
	\mathcal{C}\left(C_{k_{1}}, C_{k_{2}}\right)(\tau)=\left\{\begin{array}{ll}
			NM, & \tau=0, k_{1}=k_{2}, \\
			0, &  \tau=0, k_{1} \neq k_{2},\\
			0, & N-Z<|\tau|<N.
		\end{array}\right.
	\end{equation}
\end{definition}
When $K=M$ and $Z=N-1$, type-II ZCCS is called a $(K,K,N)$-CCC.
\subsection{EBF and Corresponding Sequences}
Let $\mathbf{x}=(x_{0},x_{1},\ldots,x_{n-1})$ such that $\mathbf{x} \in \mathbb{Z}_{p}^{n}$, where $\mathbb{Z}_p=\{0,1,\ldots,p-1\}$, for $p\ge 2$. We say $f$ is an EBF of $n$ variables $x_{0},x_{1},\ldots,x_{n-1}$ over $\mathbb{Z}_p$ if $f$ uniquely expressed as a linear combination of monomials of the above variables, where the coefficients of each monomial are taken from $\mathbb{Z}_{q}$, where $\mathbb{Z}_{q}=\{0,1,\ldots,q-1\}$ and $p|q$ \textcolor{black}{($p$ divides $q$)}. We define a sequence over the EBF $f$ as $\chi(f)=\left(\zeta_{q}^{f(0)},\zeta_{q}^{f(1)},\ldots, \zeta_{q}^{f(p^n-1)} \right)$, where $f(I)=f(I_0,I_1,\ldots,I_{n-1})$ and $\zeta_q=\exp{(2\pi \sqrt{-1}/q)}$ \cite{Rajen}, such that 
\begin{equation}\label{Eq:I_vec}
    I=\sum_{i=0}^{n-1} p^{i}I_i.
\end{equation}
\subsection{Relation with Graph and EBF}
Let $Q:\mathbb{Z}_p^n\to \mathbb{Z}_{q}$ be a EBF defined by
\begin{equation}\label{Eq:Q}
	Q(x_0,x_1,\ldots,x_{n-1})=\sum_{0\le i<j<n}h_{ij}x_ix_j,
\end{equation}
where $h_{ij}\in\mathbb{Z}_{q}$, \textcolor{black}{$p\mid q$ ($p$ divides $q$),} so that $Q$ is a quadratic form in $n$ variables, $x_0,x_1,\ldots,x_{n-1}$  over $\mathbb{Z}_{p}$. Let $\mathcal{G}(Q)$ be a labelled graph of $n$ vertices with $Q$ as given in \eqref{Eq:Q}. Label the vertices by $x_0,x_1,\ldots,x_{n-1}$ and if $h_{ij} \ne 0$, join vertices $x_i$ and $x_j$ by an edge with the label of $h_{ij}$. A graph $\mathcal{G}(Q)$ of the type defined above becomes a Hamiltonian path if  $n > 1$ and $\mathcal{G}(Q)$ has exactly $n-1$ edges of weight $\mu$, where $\mu \in \{1,2,\ldots,q-1\}$, and each vertex is connected with almost two edges ($n=1$ indicates the presence of a single vertex). For $n > 1$, a Hamiltonian path on $n$ vertices corresponds to a quadratic form of the type
\begin{equation}
	Q(x_0,x_1,\ldots,x_{n-1})=\mu \sum_{i=1}^{n-1}x_{\pi (i-1)}x_{\pi (i)},
\end{equation}
where $\pi$ is a permutation of $\{0,1,\ldots,n-1\}$. This $\pi$ is dependent on the Hamiltonian path. Now, we define an EBF over Hamiltonian path $\mathcal{G}(Q)$ by $H:\mathbb{Z}_p^n\to \mathbb{Z}_q$ as
\begin{equation}\label{Eq:EBF_Qlinear}
    H(x_0,x_1,\ldots,x_{n-1})=Q+\sum_{l=1}^{q-1}\sum_{i=0}^{n+r-1}\gamma_{i,l}x_i^l+\theta,
\end{equation}
where $\gamma_{i,l},\theta \in \mathbb{Z}_p$. For $n=1$, $H$ is a linear polynomial.
\subsection{Peak to Mean Envelope Power Ratio (PMEPR)}
Let $\mathbf{a}=(a_0,a_1,\ldots,a_{N-1})$ be a complex valued sequence of length $N$. For a multi-carrier system with $N$ subcarriers, the time domain multi-carrier signal can be written as
\begin{equation}
    s_{\mathbf{a}}(t)=\sum_{j=0}^{N-1}a_je^{2 \pi \sqrt{-1}jt},
\end{equation}
where $0\le t<1$ and the carrier spacing has been normalized to $1$ and $\mathbf{a}$ is spreaded over $N$ subcarriers. The instantaneous envelope power
of the signal is the real-valued function $P_{\mathbf{a}}(t)=|s_{\mathbf{a}}(t)|^2$. A ZCCS based MC-CDMA is given in \cite{Yubo_Comm, Rajen_ar}. From \cite{Davis}, the PMEPR is bounded by the AACF value of $\mathbf{a}$ as follows,
\begin{equation}\label{Eq:PMEPR}
    PMEPR(\mathbf{a}) \le \sum_{\tau=-N+1}^{N+1}\left|\mathcal{C}(\mathbf{a})(\tau)\right|.
\end{equation}

\section{Proposed Construction}\label{Sec:PC}
\begin{lemma}\label{Lem:sum_zero}
\textcolor{black}{Let $p,k\in \mathbb{N}$, such that $p\ge 2$, $p \nmid k$  and $\zeta_p=\exp{\frac{2\pi \sqrt{-1}}{p}}$, then
\begin{equation}
    \sum_{i=0}^{p-1} \zeta_p^{k\cdot i}=0.
\end{equation}}
\end{lemma}
\begin{proof}
    Suppose $k\ne rp$, we have $k_1=k \Mod{p}$ such that $k_1 \ne 0$ and $k_1\in \mathbb{Z}_p$. To complete the proof, we consider two cases.
    \begin{mycase}
        \case $k_1$ is relatively prime with $p$. Then
        \begin{equation}
                 \sum_{i=0}^{p-1} \zeta_p^{k\cdot i}=\sum_{i=0}^{p-1} \zeta_p^{k_1\cdot i}  =\sum_{i=0}^{p-1} \zeta_p^{i} =0, ~~ \text{sum of $p$th root of unity.}
        \end{equation}
        \case $k_1$ is not relatively prime with $p$ then there exist a $d \in \{2,3,\ldots,p-1\}$ such that $p=p_1d$, $k_1=k_2d$ and $k_2$ is relatively prime with $p_1$. Now,
       \begin{equation}
   \sum_{i=0}^{p-1} \zeta_p^{k\cdot i}=\sum_{i=0}^{p-1} \zeta_p^{k_1\cdot i}=\sum_{i=0}^{p-1} \zeta_p^{k_2d\cdot i}=\sum_{i=0}^{p-1} \zeta_{p_1}^{k_2\cdot i}=d\sum_{i=0}^{p_1-1} \zeta_{p_1}^{k_2\cdot i} =0, ~~ \text{from \textit{Case $1$}.}
        \end{equation}
    \end{mycase}
By concluding the above two cases, the proof is complete.
\end{proof}
\begin{lemma}\label{Lem:Walash_extention}
Let $\mathcal{G}(\Omega)$ be graph of $r$ isolated vertices then set of EBF $\Omega^{\delta}:\mathbb{Z}_p^r\to\mathbb{Z}_{q}$ is defined by
\begin{equation}\label{Eq:Odx}
    \Omega^{\delta}(\mathbf{x})=\frac{q}{p}\sum_{i=0}^{r-1}\lambda_i\delta_ix_i,
\end{equation}
where $\gcd(\lambda_i,p)=1$ for $i\in \{0,1,r-1\}$, $\boldsymbol{\delta}=\{\delta_0,\delta_1,\ldots,\delta_{r-1}\}$ is vector representation of $0\le \delta < p^r$, with base $p$ and $p|q$. Then the dot product between sequences $\chi(\Omega^{\delta^{1}}(\mathbf{x}))$ and $\chi(\Omega^{\delta^2}(\mathbf{x}))$ is zero for $\delta^{1}\ne \delta^{2}$.
\end{lemma}
\begin{proof}
    Let $\delta^1\ne \delta^2$, therefore, $\boldsymbol{\delta}^1$ is differ from $\boldsymbol{\delta}^2$ for at least one element. Without loss of generality, we assumed that $\boldsymbol{\delta}^1$ is differ from $\boldsymbol{\delta}^2$ for at $(u+1)$-th position of $p$-ary vector representations, for $u=0$ to $r-1$.
    \begin{mycase}
    \case Let  $\boldsymbol{\delta}^1=(\delta_0,\delta_1,\ldots,\delta^1_u,\ldots,\delta_{r-1})$ and $\boldsymbol{\delta}^2=(\delta_0,\delta_1,\ldots,\delta^2_u,\ldots,\delta_{r-1})$.
    \begin{equation}
        \Omega^{\delta^1}(\mathbf{x})-\Omega^{\delta^2}(\mathbf{x})=\frac{q\lambda_u}{p}(\delta^1_u-\delta^2_u)x_u.
    \end{equation}
    Let for an integer $I$, the $p$-ary vector representation of  $I$ be $\mathbf{I}=(I_0,I_1,\ldots,I_{r-1})$. Thus, the above can be expressed as
    \begin{equation}
        \Omega^{\delta^1}(\mathbf{I})-\Omega^{\delta^2}(\mathbf{I})=\frac{q\lambda_u}{p}(\delta^1_u-\delta^2_u)I_u.
    \end{equation}
      Now, we use this in dot product of  $\chi(\Omega^{\delta^{1}}(\mathbf{x}))$ and $\chi(\Omega^{\delta^2}(\mathbf{x}))$.
    \begin{equation}
        \begin{aligned}
            \chi(\Omega^{\delta^{1}}(\mathbf{x}))\cdot \chi(\Omega^{\delta^2}(\mathbf{x}))=&\sum_{I=0}^{p^r-1} \zeta_q^{\frac{q\lambda_u}{p}(\delta^1_u-\delta^2_u)I_u}\\
            =&p^{r-1}\sum_{I_u=0}^{p-1} \zeta_q^{\frac{q\lambda_u}{p}(\delta^1_u-\delta^2_u)I_u}\\
            =&p^{r-1}\sum_{I_u=0}^{p-1} \zeta_p^{\lambda_u(\delta^1_u-\delta^2_u)I_u}\\
            =&0,~~ \text{from Lemma \ref{Lem:sum_zero}.}
        \end{aligned}
    \end{equation}
    \case Let $\boldsymbol{\delta}^1$ be differ form  $\boldsymbol{\delta}^2$ at more than one places. Let $0\le j_0<j_1<\ldots<j_{k-1}\le r-1$, for $1< k \le r$ such that $\boldsymbol{\delta}^1$ is differ form  $\boldsymbol{\delta}^2$ at $j_0,j_1,\ldots,j_k$ places. Now
    \begin{equation}
        \Omega^{\delta^1}(\mathbf{x})-\Omega^{\delta^2}(\mathbf{x})=\frac{q\lambda_{j_u}}{p}\sum_{u=0}^{k-1}(\delta^1_{j_u}-\delta^2_{j_u})x_{j_u}.
    \end{equation}
Now, we use this in dot product of  $\chi(\Omega^{\delta^{1}}(\mathbf{x}))$ and $\chi(\Omega^{\delta^2}(\mathbf{x}))$.
    \begin{equation}
        \begin{aligned}
            \chi(\Omega^{\delta^{1}}(\mathbf{x}))\cdot \chi(\Omega^{\delta^2}(\mathbf{x}))=&\sum_{I=0}^{p^r-1} \zeta_q^{\sum_{u=0}^{k-1}\frac{q\lambda_{j_u}}{p}(\delta^1_{j_u}-\delta^2_{j_u})I_{j_u}}\\
            =& \sum_{I=0}^{p^r-1}\zeta_p^{\sum_{u=0}^{k-1}\lambda_{j_u}(\delta^1_{j_u}-\delta^2_{j_u})I_{j_u}}\\
            =& \sum_{I=0}^{p^r-1} \left( \prod_{u=0}^{k-1}\zeta_p^{\lambda_{j_u}(\delta^1_{j_u}-\delta^2_{j_u})I_{j_u}} \right)\\
            =&p^{r-k} \sum_{I_{j_u}=0}^{p-1} \zeta_p^{\lambda_{j_u}(\delta^1_{j_u}-\delta^2_{j_u})I_{j_u}}\\
            =&0,~~ \text{from Lemma \ref{Lem:sum_zero}.}
        \end{aligned}
    \end{equation}
    \end{mycase}
    The proof from the above two cases is complete.
\end{proof}
Now, we use two collections of graphs, one with Hamiltonian paths and the other with isolated vertices, to construct type-II ZCCS.
\begin{theorem}\label{Th:multi_path}
Let $\mathcal{G}(F)$ be a graph of $n+r$ vertices such that there are $r$ isolated vertices and $k$ many Hamiltonian paths from rest of $n$ vertices. For $(i+1)$-th path weight of each edges is $\frac{q\lambda_i}{p}$ and $\pi_i$ is permutation of $n_i$ vertices, where $n=n_0+n_1+\cdots+n_{k-1}$, $p|q$ and $\gcd(\lambda_i,p)=1$ for $i\in \{0,1,\ldots,k-1\}$. We label the isolated vertices as $x_0,x_1,\ldots,x_{r-1}$. Now, we define a set of EBF $F_{\beta}^{\alpha p^r+\delta}:\mathbb{Z}_p^{n+r}\to \mathbb{Z}_q$ as
\begin{equation}\label{Eq:F_ZCCS_max}
\begin{aligned}
     F_{\beta}^{\alpha p^r+\delta}(x_0,x_1,\ldots,x_{n+r-1})=&\sum_{i=0}^{k-1}\left(\frac{q\lambda_i}{p}\sum_{\nu=0}^{\nu=n_i-2}x_{\pi_i(\nu)}x_{\pi_i(\nu+1)}+\alpha_i x_{\pi_i(0)}+\beta_i x_{\pi_i(n_i-1)} \right)\\
   &+\frac{q}{p}\sum_{i=0}^{r-1}\lambda_{i+k}x_i\delta_i+\sum_{l=1}^{q-1}\sum_{i=0}^{n+r-1}\gamma_{i,l}x_i^l+\theta,
\end{aligned}
\end{equation}
where $\gcd(\lambda_i,p)=1$ for $i\in \{0,1,\ldots,k+r-1\}$, $0\le \alpha,\beta<p^k$ and $0\le \delta<p^r$. Let us define an ordered set as
\begin{equation}\label{Eq:ZCS_ZCCS_opt}
    C_{\alpha p^r+\delta}=\left[\begin{array}{c} \chi(F_0^{\alpha p^r+\delta}(\mathbf{x}))\\ \chi(F_1^{\alpha p^r+\delta}(\mathbf{x})) \\ \vdots \\ \chi(F_{p^k-1}^{\alpha p^r+\delta}(\mathbf{x})) \end{array}\right],
\end{equation}
and
\begin{equation}
    \mathbf{C}=\left\{ C_0,C_1,\ldots,C_{p^{k+r}-1} \right\}.
\end{equation}
Then $\mathbf{C}$ is a type-II $\left(p^{k+r},p^k,p^{n+r}-p^r+1,p^{n+r}\right)$-ZCCS.
\end{theorem}
\begin{proof}
  For $0\le s$ and $t < \alpha p^r+\delta$, we discuss AACF value $C_s$ obtained from \eqref{Eq:ZCS_ZCCS_opt}. Let $I=J+\tau$ such that $I,J$ and $\tau$ be positive integers and $\tau \ge p^r$. Assume
  \begin{equation}
    \begin{aligned}
         Q(x_0,x_1,\ldots,x_{n+r-1})=&\sum_{i=0}^{k-1}\left(\frac{q\lambda_i}{p}\sum_{\nu=0}^{\nu=n_i-2}x_{\pi_i(\nu)}x_{\pi_i(\nu+1)}\right)+\frac{q}{p}\sum_{i=0}^{r-1}\lambda_{i+k}x_i\delta_i\\&+\sum_{l=1}^{q-1}\sum_{i=0}^{n+r-1}\gamma_{i,l}x_i^l+\theta.
    \end{aligned}
    \end{equation}

Let $\mathbf{I}$ and $\mathbf{J}$ be $p$-ary vector representation $I$ and $J$ such that $\mathbf{I}$ and $\mathbf{J}$ have difference \textcolor{black}{at} the positions of $\{\pi_0(n_0-1),\pi_1(n_1-1), \ldots, \pi_{k-1}(n_k-1)\}$. Without loss of generality, we assume that $\mathbf{I}$ and $\mathbf{J}$ have differ \textcolor{black}{at least for}  $1\le d\le k$ places which are $\pi_{j_0}(n_{j_0}-1), \pi_{j_1}(n_{j_1}-1),\ldots, \pi_{j_{d-1}}(n_{j_{d-1}}-1)$. Then, we have
       \begin{equation}
    \begin{aligned}
        F^s_\beta(\mathbf{I})-F^t_\beta(\mathbf{J})=&\sum_{i=0}^{d-1}\frac{q\lambda_{j_i}}{p}\left((I_{\pi_{j_i}(n_{j_i}-2)}+\beta_{j_i})(I_{\pi_{j_i}(n_{j_i}-1)}-J_{\pi_{j_i}(n_{j_i}-1)})\right)+Q(\mathbf{I})-Q(\mathbf{J}),\\
        =& \sum_{i=0}^{d-1} \frac{q\lambda_{j_i}}{p}\left((I_{\pi_{j_i}(n_{j_i}-2)}+\beta_{j_i})(D^{IJ}_{\pi_{j_i}(n_{j_i}-1)})\right)+D^{IJ}_Q,
            \end{aligned}
        \end{equation}
        where $D^{IJ}_{\pi_{j_i}(n_{j_i}-1)}=I_{\pi_{j_i}(n_{j_i}-1)}-J_{\pi_{j_i}(n_{j_i}-1)}$ and $D^{IJ}_Q=Q(\mathbf{I})-Q(\mathbf{J})$. Thus,
        \begin{equation}
            \zeta_q^{ F^s_\beta(\mathbf{I})-F^t_\beta(\mathbf{J})}=\zeta_q^{D^{IJ}_Q}\zeta_p^{\sum_{i=0}^{d-1} \left((I_{\pi_{j_i}(n_{j_i}-2)}+\beta_{j_i})(\lambda_{j_i} D^{IJ}_{\pi_{j_i}(n_{j_i}-1)})\right)}.
        \end{equation}
        Taking sum over $\beta$,
        \begin{equation}\label{Eq:Fs_Ft_tau_1}
            \begin{aligned}
                \sum_{\beta=0}^{p^k-1}\zeta_q^{ F^s_\beta(\mathbf{I})-F^t_\beta(\mathbf{J})}=&\zeta_q^{D^{IJ}_Q}\sum_{\beta=0}^{p^k-1} \zeta_p^{\sum_{i=0}^{d-1} \left((I_{\pi_{j_i}(n_{j_i}-2)}+\beta_{j_i})(\lambda_{j_i} D^{IJ}_{\pi_{j_i}(n_{j_i}-1)})\right)},\\
                \sum_{\beta=0}^{p^k-1}\zeta_q^{ F^s_\beta(\mathbf{I})-F^t_\beta(\mathbf{J})}=&\zeta_q^{D^{IJ}_Q}\zeta_p^{\sum_{i=0}^{d-1} I_{\pi_{j_i}(n_{j_i}-2)} \lambda D^{IJ}_{\pi_{j_i}(n_{j_i}-1)}} \sum_{\beta=0}^{p^k-1} \zeta_p^{\sum_{i=0}^{d-1} \beta_{j_i} \lambda_{j_i} D^{IJ}_{\pi_{j_i}(n_{j_i}-1)}}.
            \end{aligned}
        \end{equation}
        Assume $\zeta_p^{\sum_{i=0}^{d-1} I_{\pi_{j_i}(n_{j_i}-2)}\lambda_{j_i} 
 D^{IJ}_{\pi_{j_i}(n_{j_i}-1)}}=\kappa_1$, then we have
        \begin{equation}\label{eq:sum_end}
            \begin{aligned}
                 \sum_{\beta=0}^{p^k-1}\zeta_q^{ F^s_\beta(\mathbf{I})-F^t_\beta(\mathbf{J})}=& \zeta_p^{D^{IJ}_Q}\kappa_1 \prod_{i=0}^{d-1}\left(\sum_{\beta=0}^{p^k-1} \zeta_p^{\beta_{j_i}\lambda_{j_i} D^{IJ}_{\pi_{j_i}(n_{j_i}-1)}}\right)\\
                 =& \zeta_p^{D^{IJ}_Q}\kappa_1 p^{k-d} \prod_{i=0}^{d-1}\left(\sum_{\beta_{j_i}=0}^{p-1} \zeta_p^{\beta_{j_i}\lambda_{j_i} D^{IJ}_{\pi_{j_i}(n_{j_i}-1)}}\right)\\
                 =& 0 ,
            \end{aligned}
        \end{equation}
\textcolor{black}{from \textit{Lemma} \ref{Lem:sum_zero}.}
        
        Let $\mathbf{I}$ and $\mathbf{J}$ \textcolor{black}{be two integers} such that $I_{\pi_i(n_i-1)}=J_{\pi_i(n_i-1)}$ for $0\le i<k$, then \textcolor{black}{find} the largest $u$ such that $I_{\pi_0(u)}\ne J_{\pi_0(u)}$, if we are unable to find such index, we again look for the largest $u$ such that $I_{\pi_1(u)}\ne J_{\pi_1(u)}$, we do this process until we found a $u$ and $v$ such that $I_{\pi_v(u)}\ne J_{\pi_v(u)}$. Let us consider $I_{\pi_v(u)}\ne J_{\pi_v(u)}$ and for all $u^\prime >u$, we have $I_{\pi_v(u^\prime)}= J_{\pi_v(u^\prime )}$. Let us consider $I^\kappa$ which is the integer having vector representation with base $p$ differ from $\mathbf{I}$ only at the position $\pi_v(u+1)$, i.e.,
        \begin{equation}\label{eq:II_k}
            (I_0,I_1,\ldots,I_{\pi_v(u+1)}-\kappa,\ldots,I_{n+r-1}),
        \end{equation}
        where $\kappa\in \{1,2,\ldots,p-1\}$. Similarly, $J^\kappa$ has vector representation with base $p$ as
        \begin{equation}\label{eq:JJ_k}
            (J_0,J_1,\ldots,J_{\pi_v(u+1)}-\kappa,\ldots,J_{n+r-1}).
        \end{equation}
        \textcolor{black}{As $J=I+\tau$ it implies $J^\kappa=I^\kappa+\tau$.}
From \eqref{Eq:F_ZCCS_max} and \eqref{eq:II_k},
\begin{equation}\label{eq:Fs:IkI}
    F^s_{\beta}(\mathbf{I}^\kappa)- F^s_{\beta}(\mathbf{I})=\frac{q\lambda_v}{p}(-\kappa)\left(I_{(\pi_v(u)}+I_{\pi_v(u+2)}+\sum_{l=1}^{q-1} \gamma_{\pi_v(u),l}\kappa^l \right).
\end{equation}
For $u=n-2$, \eqref{eq:Fs:IkI} changes to
\begin{equation}
    F^s_{\beta}(\mathbf{I}^\kappa)- F^s_{\beta}(\mathbf{I})=\frac{q\lambda_v}{p}(-\kappa)\left(I_{\pi_v(u)}+\sum_{l=1}^{q-1} \gamma_{\pi_v(u),l}\kappa^l \right).
\end{equation}
Similarly,
\begin{equation}\label{eq:Fs:JkJ}
    F^t_{\beta}(\mathbf{J}^\kappa)- F^t_{\beta}(\mathbf{J})=\frac{q\lambda_v}{p}(-\kappa)\left(J_{(\pi_v(u)}+J_{\pi_v(u+2)}+\sum_{l=1}^{q-1} \gamma_{\pi_v(u),l}\kappa^l \right).
\end{equation}
Subtracting \eqref{eq:Fs:JkJ} from \eqref{eq:Fs:IkI} and put $I_{\pi_v(u+2)}=J_{\pi_v(u+2)}$, we have
\begin{equation}
\begin{aligned}
    F^s_{\beta}(\mathbf{I}^\kappa)- F^t_{\beta}(\mathbf{J}^\kappa)-(F^s_{\beta}(\mathbf{I})-F^t_{\beta}(\mathbf{J}))=&\frac{q\lambda_v}{p}\kappa(I_{\pi_v(u)}-J_{\pi_v(u)})\\
    F^s_{\beta}(\mathbf{I}^\kappa)- F^t_{\beta}(\mathbf{J}^\kappa)=&F^s_{\beta}(\mathbf{I})-F^t_{\beta}(\mathbf{J})+\frac{q\lambda_v}{p}\kappa D^{IJ}_{\pi_v(u)}\\
    \zeta_q^{F^s_{\beta}(\mathbf{I}^\kappa)- F^t_{\beta}(\mathbf{J}^\kappa)}=&\zeta_q^{F^s_{\beta}(\mathbf{I})-F^t_{\beta}(\mathbf{J})}\zeta_q^{\frac{q\lambda_v}{p}\kappa D^{IJ}_{\pi_v(u)}}.
\end{aligned} 
\end{equation}
Summing over $\kappa$,
\begin{equation}
    \begin{aligned}
        \sum_{\kappa=1}^{p-1}\zeta_q^{F^s_{\beta}(\mathbf{I}^\kappa)- F^t_{\beta}(\mathbf{J}^\kappa)}=&\zeta_q^{F^s_{\beta}(\mathbf{I})-F^t_{\beta}(\mathbf{J})} \sum_{\kappa=1}^{p-1}\zeta_q^{\frac{q\lambda_v}{p}\kappa D^{IJ}_{\pi_v(u)}}\\
        =&\zeta_q^{F^s_{\beta}(\mathbf{I})-F^t_{\beta}(\mathbf{J})} \sum_{\kappa=1}^{p-1}\zeta_p^{\kappa \lambda_v D^{IJ}_{\pi_v(u)}}\\
         =&-\zeta_q^{F^s_{\beta}(\mathbf{I})-F^t_{\beta}(\mathbf{J})}.
    \end{aligned}
\end{equation}
Therefore,
\begin{equation}
    \sum_{\kappa=1}^{p-1}\zeta_q^{F^s_{\beta}(\mathbf{I}^\kappa)-F^t_\beta(\mathbf{J}^\kappa)}+\zeta_q^{ F^s_\beta (\mathbf{I})-F^t_\beta (\mathbf{J})}=0, ~\text{for all $\beta$.}
\end{equation}
Now summing over $\beta$,
\begin{equation}\label{eq:sum_kappa}
    \sum_{\beta=0}^{p^k-1}\sum_{\kappa=1}^{p-1}\zeta_q^{F^s_{\beta}(\mathbf{I}^\kappa)-F^t_\beta(\mathbf{J}^\kappa)}+\sum_{\beta=0}^{p^k-1}\zeta_q^{ F^s_\beta (\mathbf{I})-F^t_\beta (\mathbf{J})}=0.
\end{equation}
$(I^\kappa,J^\kappa)$, for $\kappa \in \{1,2,\ldots,p-1\}$ contribute to $\mathcal{C}(C_s,C_t)(\tau)$, whenever $(I,J)$ contribute to $\mathcal{C}(C_s,C_t)(\tau)$. From \eqref{eq:sum_end} and \eqref{eq:sum_kappa}, we have
\begin{equation}\label{eq:corr_zcz_zero}
    \mathcal{C}(C_s,C_t)(\tau)=0,\text{~ for all $s,t$ and $|\tau|>p^r-1$.}
\end{equation}
\textcolor{black}{Now, we process for ACCF between two codes at time shift zero; for that, let $0\le s\ne t\le p^{k+r}$.} First we represent $s$ and $t$ in the form of $\alpha p^r+ \delta$, where $0\le \delta<p^r$ and $0\le \alpha <p^n$. As $s\ne t$, at least one of the terms from $p$-ary representation of either $\alpha$ or $\delta$ is not the same. Let $\alpha_{i_0},\alpha_{i_2},\ldots,\alpha_{i_{a-1}}$ and $\delta_{j_0},\delta_{j_1},\ldots,\delta_{j_{b-1}}$, such that $0\le a < n$, $0\le b < r$ and $a+b\ge 0$. Then, we have
\begin{equation}
\begin{aligned}
        F^s_{\beta}(\mathbf{I})-F^t_{\beta}(\mathbf{I})=&\frac{q}{p}\left( \sum_{k=0}^{a-1} \lambda_{i_k}I_{\pi_{i_k}(0)} \alpha_{i_k}+\sum_{l=0}^{b-1} \lambda_{j_l} I_{n+j_l} \delta_{j_l} \right),\\
        \zeta_q^{F^s_{\beta}(\mathbf{I})-F^t_{\beta}(\mathbf{I})}=& \zeta_q^{\frac{q}{p} \left( \sum_{k=0}^{a-1} \lambda_{i_k}I_{\pi_{i_k}(0)}\alpha_{i_k}+\sum_{l=0}^{b-1} \lambda_{j_l} I_{n+j_l} \delta_{j_l} \right)}\\
        =& \zeta_p^{\lambda \left( \sum_{k=0}^{a-1} \lambda_{i_k}I_{\pi_{i_k}(0)}\alpha_{i_k}+\sum_{l=0}^{b-1} \lambda_{j_l} I_{n+j_l} \delta_{j_l} \right)}.
\end{aligned}
\end{equation}
Taking sum over $I$,
\begin{equation}
    \sum_{I=0}^{p^{n+r}-1} \zeta_q^{F^s_{\beta}(\mathbf{I})-F^t_{\beta}(\mathbf{I})}= \sum_{I=0}^{p^{n+r}-1} \zeta_p^{\lambda \left(\sum_{k=0}^{a-1} I_{i_k}\alpha_{i_k}+\sum_{l=0}^{b-1} I_{n+j_l}\delta_{j_l}\right)}=0,
\end{equation}
from \textit{Lemma} \ref{Lem:Walash_extention} for all $\beta$. Now taking sum over $\beta$
\begin{equation}\label{eq:orthogonal}
    \begin{aligned}
        \sum_{\beta=0}^{p^k-1} \sum_{I=0}^{p^{n+r}-1} \zeta_q^{F^s_{\beta}(\mathbf{I})-F^t_{\beta}(\mathbf{I})}&=0,\\
        \implies \mathcal{C}(C_s,C_t)(0)&=0.
    \end{aligned}
\end{equation}
From \eqref{eq:corr_zcz_zero} and \eqref{eq:orthogonal}, we have
\begin{equation}
    \mathcal{C}(C_s,C_t)(\tau)=\left\{\begin{array}{ll}
			0, &  \tau=0,~ s \neq t,\\
			0, & p^r\ge |\tau|.
		\end{array}\right.
\end{equation}
This completes the proof.
\end{proof}

\begin{remark}\label{Rm:CCC}
\textit{Theorem}  \ref{Th:multi_path}, without any isolated vertices $(r=0)$, provides $(p^k,p^k,p^n)$-CCC.
\end{remark}

In \textit{Theorem} \ref{Th:multi_path}, column sequence PMEPR upper bound is $p^r$, which can be reduced by adding a suitable function. In \cite{SarkarLCOM}, the authors have suggested \textcolor{black}{reducing} the column sequence PMEPR with a suitable function. If a sequence is part of the Golay complementary set with $p$ sequences, its PMEPR upper bound is $p$. Every column of a code obtained from \textit{Theorem} \ref{Th:multi_path} can be represented by a linear function of $(\beta_0,\beta_1,\ldots,\beta_{k-1})$, and if we add some function of $(\beta_0,\beta_1,\ldots,\beta_{k-1})$ in \eqref{Eq:F_ZCCS_max}, it for any row of the code, any function of $(\beta_0,\beta_1,\ldots,\beta_{k-1})$ is a constant, which does not affect the AACF/ACCF value of the codes. Therefore, we propose a remark that ensures low column sequence PMEPR.
\begin{remark}
    In \textit{Theorem} \ref{Th:multi_path}, replace \eqref{Eq:F_ZCCS_max} by 
    \begin{equation}
    F_{\beta}^{\alpha p^r+\delta}= f_{\beta}^{\alpha} + \Omega^{\delta}+\frac{q}{p}\sum_{\nu=1}^{k-1}\beta_{\nu-1}\beta_\nu,
\end{equation}
 for $0\le \alpha,\beta<p^k$ and $0\le \delta<p^r$ and $(\beta_0,\beta_1,\ldots,\beta_{k-1})$ is $p$-ary vector representation of $\beta$. Then the obtained type-II ZCCS have column sequence PMEPR bounded above by $p$.
\end{remark}
\begin{remark}
As $p=2$ and $q=2$, we get type-II binary $(2^{k+r},2^k,2^{n+r}-2^r+1,2^{n+r})$-ZCCS.
\end{remark}

\begin{example}\label{Ex:16codes}
    Let $\mathcal{G}(F)$ be a graph of two Hamiltonian paths, each having two vertices denoted by $x_0,x_1$ and $x_2,x_3$, respectively and $\mathcal{G}(\Omega)$ be an isolated vertex denoted by $x_4$, as shown in Fig. \ref{fig:graph}.
 \begin{figure}[!ht]
   \centering
   \begin{tikzpicture}[node distance={22.5mm}, thick, main/.style = {draw, circle}]
\node[main] (1) {$x_0$};
\node[main] (2) [above right of=1] {$x_1$};
\node[main] (3) [below right of=1] {$x_3$};
\node[main] (4) [above right of=3] {$x_2$}; 
\node[main] (5) [below right of=4] {$x_4$};
\draw (4) -- (2);
\draw (3) -- (5);
\end{tikzpicture} 
\centering
\caption{Graph of two Hamiltonian paths and one isolated vertex.}
   \label{fig:graph}
 \end{figure}
Now, define EBFs $F_{\beta}^{2\alpha+\delta}:\mathbb{Z}_2^5\to \mathbb{Z}_2$ as follows,
    \begin{equation}\label{Eq:example}
F_{\beta}^{2\alpha+\delta}=x_1x_2+x_3x_4+\alpha_0x_1+\alpha_1x_3+\beta_0x_2+\beta_1x_4+\delta_0 x_0+\beta_0\beta_1.
    \end{equation}
    Using \eqref{Eq:example} in \textit{Theorem} \ref{Th:multi_path}, we get
    \begin{equation}
C_{2\alpha+\delta}=\left[\begin{array}{c} \chi(F_{0}^{2\alpha+\delta}(\mathbf{x}))\\ \chi(F_1^{2\alpha+\delta}(\mathbf{x})) \\\chi(F_{2}^{2\alpha+\delta}(\mathbf{x})) \\ \chi(F_{3}^{2\alpha+\delta}(\mathbf{x})) \end{array}\right],
    \end{equation}
for $0\le \alpha,\beta < 4$ and $0\le \delta<2$. Then $\mathbf{C}=\{C_0,C_1,\ldots,C_{8}\}$ are given below ($+$ represents $1$ and $-$ represents $-1$) is a binary type-II $(8,4,31,32)$-ZCCS. Using \eqref{Eq:PMEPR}, we get column sequence PMEPR, which is upper bounded by $2$ for the obtained codes. 
 \begin{align*}
     C_0=&\begin{bmatrix}
        +     +     +     +     +     +    -    -     +     +     +     +     +     +    -    -     +     +     +     +     +     +    -    -    -    -    -    -    -    -     +     +\\
     +     +     +     +    -    -     +     +     +     +     +     +    -    -     +     +     +     +     +     +    -    -     +     +    -    -    -    -     +     +    -    -\\
     +     +     +     +     +     +    -    -     +     +     +     +     +     +    -    -    -    -    -    -    -    -     +     +     +     +     +     +     +     +    -    -\\
    -    -    -    -     +     +    -    -    -    -    -    -     +     +    -    -     +     +     +     +    -    -     +     +    -    -    -    -     +     +    -    -
     \end{bmatrix},\\
     C_1=&\begin{bmatrix}
      +    -     +    -     +    -    -     +     +    -     +    -     +    -    -     +     +    -     +    -     +    -    -     +    -     +    -     +    -     +     +    -\\
     +    -     +    -    -     +     +    -     +    -     +    -    -     +     +    -     +    -     +    -    -     +     +    -    -     +    -     +     +    -    -     +\\
     +    -     +    -     +    -    -     +     +    -     +    -     +    -    -     +    -     +    -     +    -     +     +    -     +    -     +    -     +    -    -     +\\
    -     +    -     +     +    -    -     +    -     +    -     +     +    -    -     +     +    -     +    -    -     +     +    -    -     +    -     +     +    -    -     +
     \end{bmatrix},\\
     C_2=&\begin{bmatrix}
     +     +    -    -     +     +     +     +     +     +    -    -     +     +     +     +     +     +    -    -     +     +     +     +    -    -     +     +    -    -    -    -\\
     +     +    -    -    -    -    -    -     +     +    -    -    -    -    -    -     +     +    -    -    -    -    -    -    -    -     +     +     +     +     +     +\\
     +     +    -    -     +     +     +     +     +     +    -    -     +     +     +     +    -    -     +     +    -    -    -    -     +     +    -    -     +     +     +     +\\
    -    -     +     +     +     +     +     +    -    -     +     +     +     +     +     +     +     +    -    -    -    -    -    -    -    -     +     +     +     +     +     +
     \end{bmatrix},\\
     C_3=&\begin{bmatrix}
     +    -    -     +     +    -     +    -     +    -    -     +     +    -     +    -     +    -    -     +     +    -     +    -    -     +     +    -    -     +    -     +\\
     +    -    -     +    -     +    -     +     +    -    -     +    -     +    -     +     +    -    -     +    -     +    -     +    -     +     +    -     +    -     +    -\\
     +    -    -     +     +    -     +    -     +    -    -     +     +    -     +    -    -     +     +    -    -     +    -     +     +    -    -     +     +    -     +    -\\
    -     +     +    -     +    -     +    -    -     +     +    -     +    -     +    -     +    -    -     +    -     +    -     +    -     +     +    -     +    -     +    -
     \end{bmatrix},\\
     C_4=&\begin{bmatrix}
           +     +     +     +     +     +    -    -    -    -    -    -    -    -     +     +     +     +     +     +     +     +    -    -     +     +     +     +     +     +    -    -\\
     +     +     +     +    -    -     +     +    -    -    -    -     +     +    -    -     +     +     +     +    -    -     +     +     +     +     +     +    -    -     +     +\\
     +     +     +     +     +     +    -    -    -    -    -    -    -    -     +     +    -    -    -    -    -    -     +     +    -    -    -    -    -    -     +     +\\
    -    -    -    -     +     +    -    -     +     +     +     +    -    -     +     +     +     +     +     +    -    -     +     +     +     +     +     +    -    -     +     +
     \end{bmatrix},\\
     C_5=&\begin{bmatrix}
     +    -     +    -     +    -    -     +    -     +    -     +    -     +     +    -     +    -     +    -     +    -    -     +     +    -     +    -     +    -    -     +\\
     +    -     +    -    -     +     +    -    -     +    -     +     +    -    -     +     +    -     +    -    -     +     +    -     +    -     +    -    -     +     +    -\\
     +    -     +    -     +    -    -     +    -     +    -     +    -     +     +    -    -     +    -     +    -     +     +    -    -     +    -     +    -     +     +    -\\
    -     +    -     +     +    -    -     +     +    -     +    -    -     +     +    -     +    -     +    -    -     +     +    -     +    -     +    -    -     +     +    -
     \end{bmatrix},\\
     C_6=&\begin{bmatrix}
     +     +    -    -     +     +     +     +    -    -     +     +    -    -    -    -     +     +    -    -     +     +     +     +     +     +    -    -     +     +     +     +\\
     +     +    -    -    -    -    -    -    -    -     +     +     +     +     +     +     +     +    -    -    -    -    -    -     +     +    -    -    -    -    -    -\\
     +     +    -    -     +     +     +     +    -    -     +     +    -    -    -    -    -    -     +     +    -    -    -    -    -    -     +     +    -    -    -    -\\
    -    -     +     +     +     +     +     +     +     +    -    -    -    -    -    -     +     +    -    -    -    -    -    -     +     +    -    -    -    -    -    -
     \end{bmatrix},\\
     C_7=&\begin{bmatrix}
     +    -    -     +     +    -     +    -    -     +     +    -    -     +    -     +     +    -    -     +     +    -     +    -     +    -    -     +     +    -     +    -\\
     +    -    -     +    -     +    -     +    -     +     +    -     +    -     +    -     +    -    -     +    -     +    -     +     +    -    -     +    -     +    -     +\\
     +    -    -     +     +    -     +    -    -     +     +    -    -     +    -     +    -     +     +    -    -     +    -     +    -     +     +    -    -     +    -     +\\
    -     +     +    -     +    -     +    -     +    -    -     +    -     +    -     +     +    -    -     +    -     +    -     +     +    -    -     +    -     +    -     +
     \end{bmatrix}.
 \end{align*}
\end{example}
In type-I ZCCS with $4$ sequences of length $32$, the number of codes can be $8$ only if ZCZ width is less than $16$ due to the bound $K \leq M \left\lfloor \frac{N}{Z}\right\rfloor$. However, \textit{Example} \ref{Ex:16codes} provides $8$ codes with the ZCZ width $31$.

In Table \ref{Tab:ZCCS}, we provide some available parameters for type-I ZCCS and type-II ZCCS, where it can be observed that type-II ZCCS provides a larger number of codes compared to type-I ZCCS for a specific ZCZ width.
\section{Comparison}\label{sec:compar}
\textcolor{black}{For $r=0$, \textit{Theorem} \ref{Th:multi_path} produces $(p^k,p^k,p^n)$-CCC, therefore, CCC presented in \cite{Shen_SNC} can be seen as a special case of the proposed construction.} Similarly, for $p=2$ and $r=0$, the proposed construction produces $(2^k,2^k,2^n)$-CCC. Hence, \cite{Rat} also appears as a special case of our proposed construction. Further, for $p=2$, $r=0$ and $k=0$, the proposed construction produces a GCP of length in the form of $2^n$, with its complementary mate. Hence, \cite{Davis}  also appears as a special case of the proposed construction.

 We are providing a table to compare type-II ZCCS with the existing ZCCSs.
\begin{table}[!ht]
		\centering
			\caption{Comparison with some available ZCCS}\label{Tab:ZCCS}
		\begin{tabular}{@{}lllll@{}}
			\toprule
			 Source & Method & ZCCS & $\left(K,M,Z,N\right)$ & Constraint \\
		\midrule
	\cite{SarkarLCOM}& GBF & Type-I & $\left(2^{k+p},2^k,2^{m-p},2^m\right)$& $p\le m$, $k+1\ge m$, $k\ge 0$\\
         \cite{Sarkarpseudo}& pseudo BF &Type-I & $\left(p2^{k+1},2^{k+1},2^{m},p2^m\right)$& $p$ is prime, $k+1\le m$  \\
         \cite{Ghosh} & GBF & Type-I & $\left(P2^{k+1},2^{k+1},2^m,P2^m\right)$& $P\ge 2$, $k+1\le m$  \\
                \cite{Rajen_ar} & Kronecker & Type-II & $\left(rK,K,NP-P+1,NP\right)$& $r\le P$  \\
                &  product & & & $K\times N$ code size of a CCC\\
         Th. \ref{Th:multi_path} & EBF & Type-II & $\left(P^{k+r},P^k,P^{n+r}-P^r+1,P^{n+r}\right)$& $P\ge 2$, $k\le n$, $r,k\ge 0$\\
         \botrule
			\end{tabular}
	\end{table}
\section{Conclusion}
\label{Sec:Con}
In order to construct a type-II ZCCS, we connected EBF to a graph that mainly consists of Hamiltonian paths and isolated vertices. A construction for type-II $(p^{k+r},p^{k},p^{r+n}-p^r+1,p^{r+n})$-ZCCS has been proposed using the EBF corresponding to the graph of $k$ collection of Hamiltonian paths and $r$ isolated vertices such that $p\ge 2$, $r\ge 0$ and $0<k<n$. The column sequence PMEPR is also bounded by $p$. Furthermore, the proposed construction can generate $(p,p,p)$-CCC, for any integral value of $p$. The proposed type-II $(K,M,Z,N)$-ZCCS follows the code bound $K=M(N-Z+1)$. The works in \cite{Davis, Rat, Shen_SNC} are \textcolor{black}{special cases} of the proposed construction.

\bibliography{sn-bibliography}

\end{document}